\documentclass[conference]{IEEEtran}
\IEEEoverridecommandlockouts

\usepackage{graphicx}
\usepackage{amssymb}
\usepackage{amsmath}
\usepackage{cite}
\usepackage{subfigure}
\usepackage{mathrsfs}
\usepackage[displaymath,mathlines]{lineno}
\usepackage{color}
\usepackage{tabulary}
\usepackage{multirow}
\usepackage{algpseudocode}
\usepackage{algorithm,algpseudocode}
\usepackage{algorithmicx}
\usepackage{pbox}
\usepackage{multicol}
\usepackage{lipsum}
\usepackage{amsthm}
\usepackage{relsize}
\usepackage{lipsum}
\usepackage{epstopdf}
\usepackage{enumitem}

\newcommand{\doublewidetilde}[1]{{%
		\mathpalette\double@widetilde{#1}%
	}}
	
\makeatletter

\def\BState{\State\hskip-\ALG@thistlm}
\makeatother

\newtheorem{definition}{Definition}

\newtheorem{lemma}{Lemma}
\newtheorem{corollary}{Corollary}

\newtheorem{remark}{Remark}

\newcounter{eqnback}
\newcounter{eqncnt}

\begin{document}
%
\title{Uplink Power Control in Cellular Massive MIMO Systems: Coping With the Congestion Issue}

\author{\IEEEauthorblockN{Trinh Van Chien$^{*}$, Emil Bj\"{o}rnson$^{*}$, and Hien Quoc Ngo$^{\dagger}$}
	\IEEEauthorblockA{$^{*}$Department of Electrical
		Engineering (ISY), Link\"{o}ping University, SE-581 83 Link\"{o}ping, Sweden\\
		$^{\dagger}$School of Electronics, Electrical Engineering and Computer Science, Queen's University Belfast, Belfast, UK \\
		trinh.van.chien@liu.se, emil.bjornson@liu.se, and hien.ngo@qub.ac.uk}
	\thanks{This paper was supported by ELLIIT and CENIIT. The work of H. Q. Ngo was supported by the UK Research and Innovation Future Leaders Fellowships under Grant MR/S017666/1.}
}

\maketitle

\begin{abstract}
One main goal of 5G-and-beyond systems is to simultaneously serve many users, each having a requested spectral efficiency (SE), in an energy-efficient way. The network capacity cannot always satisfy all the SE requirements, for example, when some users have bad channel conditions, especially happening in a cellular topology, and therefore congestion can happen. By considering both the pilot and data powers in the uplink transmission as optimization variables, this paper formulates and solves an energy-efficiency problem for cellular Massive MIMO (Multiple Input Multiple Output) systems that can handle the congestion issue. New algorithms based on the alternating optimization approach are proposed to obtain a fixed-point solution. Numerical results manifest that the proposed algorithms can provide the demanded SEs to many users even when the congestion happens.
\end{abstract}

\IEEEpeerreviewmaketitle

\section{Introduction}
\vspace*{-0.15cm}
Massive MIMO can provide large improvements in spectral efficiency (SE) and energy efficiency in the next-generation cellular systems by enabling multiple users to share the same time and frequency resource \cite{massivemimobook}. Different from current wireless systems, each base station (BS) in Massive MIMO systems can use a simple linear processing scheme such as maximum ratio combining (MRC) or regularized zero forcing to obtain good performance. Due to a large number of users, there will be plenty of inter-user interference that needs to be dealt with by appropriate resource allocation. In Massive MIMO, some of the resource allocation tasks (e.g., power control) can be done on the large-scale fading time scale instead of the conventional small-scale fading time scale, thanks to the channel hardening property \cite{massivemimobook}. This might make it feasible to implement advanced resource allocation algorithms in practical Massive MIMO systems.

Instead of scheduling users in time based on their long-term rate requirements, Massive MIMO systems can serve all users simultaneously and set the per-user SEs equal to that required by each user application \cite{zhang2014heterogeneous}. When there are given SE requirements, the resource allocation can be optimized to satisfy them with maximum energy efficiency. This was done in \cite{Chien2016b} where the total transmit power was minimized. However, since the users were randomly distributed, around $30\%$ of realizations of the user locations resulted in infeasible optimization problems due to poor channel conditions for some users \cite{Chien2016b, senel2017joint}. This congestion issue was handled in \cite{Chien2016b, Marzetta2016a} by ignoring the SE constraints and instead maximizing the minimum SE of the users, which is an entirely different problem that might lead to not satisfying any of the SE requirements.
 To the best of our knowledge, there is no previous work that deals with congestion due to infeasible SE constraints in Massive MIMO. 

In conventional multi-user MIMO communications, the congestion issue was discussed in \cite{Stridh2006a, sung2005generalized, pang2008distributed} and references therein.
A primal-dual decomposition was used in \cite{Stridh2006a} to iteratively identify and remove the user that interfers the most with the other users until the remaining SE requirements can be satisfied. While there were no power constraints in \cite{Stridh2006a}, the case with power constraints was considered in \cite{sung2005generalized, pang2008distributed}. These papers developed policies to decrease the requested SE of the users with poor channel conditions. 

These previous works consider perfect channel state information (CSI), so the impacts of imperfect CSI and pilot contamination are not considered, and the pilot powers are not optimized. Additionally, the SE is computed based on the instantaneous channel realizations (small-scale fading) which require the optimization problems to be solved often to combat the small-scale fading. In contrast, this paper tackles the congestion issue in cellular Massive MIMO systems by considering the ergodic SE. 
 Our main goal is to serve as many users  as possible with their requested SEs. We formulate a total uplink power minimization problem where both the pilot and data powers are optimization variables. To cope with the congestion issue originating from users with poor channel conditions, we propose a new algorithm that exploits the alternating optimization method with low computational complexity. We show numerically how the system can operate effectively also under congestion.
 

\textit{Notations}: The bold upper-case and lower-case letters denote matrices and vectors, respectively. $\mathbb{E} \{ \cdot \}$ is the expectation of a random variable. The transpose and Hermitian transpose are denoted as $(\cdot)^T$ and $(\cdot)^H$, respectively. Finally, $\mathcal{CN}(\cdot, \cdot)$ is the circularly symmetric Gaussian distribution.

\setcounter{eqnback}{\value{equation}} \setcounter{equation}{4}
\begin{figure*}
	\begin{equation} \label{eq:SINRlk}
	\mathrm{SINR}_{lk} = \frac{M (\beta_{lk}^l)^2 \tau_p p_{lk} \hat{p}_{lk} }{\left( \tau_p \sum_{l' =1}^L \beta_{l'k}^l  \hat{p}_{l'k}  + \sigma^2 \right)\left( \sum_{l'=1}^L \sum_{k'=1}^K   p_{l'k'} \beta_{l'k'}^l + \sigma^2 \right)  + M \tau_p \sum_{l' = 1, l' \neq l }^L  (\beta_{l'k}^l)^2  p_{l'k} \hat{p}_{l'k} }
	\end{equation} \vspace*{-0.2cm} 
	\hrule
	\vspace*{-0.5cm}
\end{figure*}
\setcounter{eqncnt}{\value{equation}}
\setcounter{equation}{\value{eqnback}}

\section{System Model} \label{Section: System Model}
\vspace*{-0.15cm}
We consider a multi-cell Massive MIMO system with $L$ cells, each having a BS equipped with $M$ antennas and serving $K$ users. Even though the radio channels vary over time and frequency, we assume that the radio resources are divided into coherence intervals of $\tau_c$ symbols where the channels are static and frequency flat. The channel between user~$k'$ in cell~$l'$ and BS~$l$, $\mathbf{h}_{l'k'}^l \in \mathbb{C}^{M},$ follows uncorrelated Rayleigh fading:
$\mathbf{h}_{l'k'}^l \sim \mathcal{CN}(\mathbf{0}, \beta_{l'k'}^l \mathbf{I}_{M}),$
where $\beta_{l'k'}^l$ is the large-scale fading coefficient. Each BS will estimate the realizations of the channels from itself to the intra-cell users in the uplink pilot training phase, which requires an overhead of $\tau_p \geq K$ symbols in each coherence interval.

\vspace*{-0.15cm}
\subsection{Uplink Pilot Training}
\vspace*{-0.15cm}
We assume that a set of $K$ orthogonal pilot signals $\{ \pmb{\psi}_{1}, \ldots, \pmb{\psi}_{K} \} $ is reused in every cell of the system with $\pmb{\psi}_{k} \in \mathbb{C}^{\tau_p}$ and $ \|\pmb{\psi}_{k} \|^2 = \tau_p$. Without loss of generality, we assume that every user~$k$ in each cell shares the pilot signal $\pmb{\psi}_{k}$. The received pilot signal, $\mathbf{Y}_l \in \mathbb{C}^{M \times \tau_p}$, at BS~$l$ is
\begin{equation} \label{eq:ReceivedPilotSig}
\mathbf{Y}_l = \sum_{l'=1}^L \sum_{k'=1}^K \sqrt{\hat{p}_{l'k'}}\mathbf{h}_{l'k'}^l \pmb{\psi}_{k'}^H + \mathbf{N}_l,
\end{equation}
where $\hat{p}_{l'k'}$ is the pilot power used by user~$k'$ in cell~$l'$ and $\mathbf{N}_l \in \mathbb{C}^{M \times \tau_p }$ is additive receiver noise with the elements independently distributed as $\mathcal{CN}(0, \sigma^2)$ with variance $\sigma^2$. The minimum mean square error estimate of $\mathbf{h}_{l'k'}^l$ can be obtained from \eqref{eq:ReceivedPilotSig} as \cite{Kay1993a}
\begin{equation}
\hat{\mathbf{h}}_{l'k'}^l = \frac{\sqrt{\hat{p}_{lk'}} \beta_{l'k'}^l}{\tau_p \sum_{l''=1}^{L} \hat{p}_{l'' k'} \beta_{l''k'}^l+\sigma^2} \mathbf{Y}_l \pmb{\psi}_{k'}.
\end{equation} 
These channel estimates will be used for linear combining of the uplink data signals.
\subsection{Uplink Data Transmission}
\vspace*{-0.15cm}
In the uplink data transmission phase, user~$k'$ in cell~$l'$ is transmitting a data symbol $s_{l'k'}$ with $\mathbb{E} \{ |s_{l'k'}|^2\} =1$. The received signal at BS~$l$, $\mathbf{y}_{l} \in \mathbb{C}^{M}$, is the superposition of transmitted signals from the $KL$ users in the system:
\begin{equation}
\mathbf{y}_{l} = \sum_{l'=1}^L \sum_{k'=1}^K \sqrt{p_{l'k'}} \mathbf{h}_{l'k'}^l s_{l'k'} + \mathbf{n}_{l},
\end{equation}
where $p_{l'k'}$ is the data power allocated by user~$k'$ in cell~$l'$ and $\mathbf{n}_{l} \in \mathbb{C}^{M}$ is the additive noise distributed as $\mathcal{CN} \left(\mathbf{0}, \sigma^2 \mathbf{I}_{M} \right)$. By utilizing MRC, a closed-form expression for the uplink ergodic SE of user~$k$ in cell~$l$ is  \cite[Corollary~$1$]{Chien2017a}:
\begin{equation} \label{eq:DLRate}
R_{lk} = \left(1 - \frac{\tau_p}{\tau_c} \right) \log_2 \left( 1 + \mathrm{SINR}_{lk}\right),
\end{equation}
where the effective signal-to-interference-and-noise ratio (SINR), $\mathrm{SINR}_{lk}$, is given by \eqref{eq:SINRlk}, shown at the top of this page.
The closed-form SE in \eqref{eq:DLRate} enables us to optimize the resource allocation without resorting to Monte-Carlo simulations. In particular, we will formulate an uplink power minimization problem that jointly optimizes both the pilot and data powers.

\section{Uplink Energy Consumption Minimization}
\vspace*{-0.15cm}
This section describes an uplink energy consumption minimization problem and then demonstrates its infeasibility for certain user locations, where users may not meet the requested SE due to interference and limited power budgets. We assume that user $k$ in cell $l$ requests an SE $\xi_{lk} >0$ and has a maximum transmit power budget of $P_{\max,lk}>0$.
\vspace*{-0.15cm}

\subsection{Problem Formulation}
\vspace*{-0.15cm}
Since $\tau_p$ pilots and $\tau_c-\tau_p$ data symbols are transmitted per coherence interval, the total energy consumption in cell~$l$ per sent for the pilot and data symbols as
\setcounter{eqnback}{\value{equation}} \setcounter{equation}{5}
\begin{equation} \label{eq:Etransl}
E_{l} =  \tau_p \sum_{k=1}^{K} \hat{p}_{lk} + (\tau_c - \tau_p) \sum_{k=1}^{K} p_{lk}.
\end{equation}
We minimize the total energy consumed of all cells subject to the requested SEs and limited power per symbol as
\begin{equation} \label{eq:General-Form-Optimization}
\begin{aligned}
& \underset{ \substack{\{\hat{p}_{lk}, p_{lk} \in  \mathbb{R}_{+} \} } }{\textrm{minimize}}
& & \sum_{ l=1 }^{L} E_{l} \\
& \,\,\textrm{subject to}
& &  R_{lk} \geq \xi_{lk},\quad \forall l,k ,\\
& &&  \hat{p}_{lk} \leq P_{\mathrm{max}, lk }, \quad \forall l,k,\\
&&& p_{lk} \leq  P_{\mathrm{max}, lk }, \quad \forall l,k,
\end{aligned}
\end{equation}
where $\xi_{lk}$ is the requested SE of user $k$ in cell $l$ and the ergodic SE in \eqref{eq:DLRate} must be larger or equal to it; $P_{\max,lk}$ is the maximum power that user~$k$ in cell~$l$ can supply to pilot and data symbols. By setting $\hat{\xi}_{lk} = 2^{\xi_{lk} \tau_c /(\tau_c - \tau_p)} -1 $, problem~\eqref{eq:General-Form-Optimization} is converted from having SEs constraints to the corresponding formulation with SINR constraints:
\begin{subequations} \label{eq:Opt1}
	\begin{alignat}{2}
	& \underset{ \substack{\{ \hat{p}_{lk}, p_{lk} \in \mathbb{R}_{+} \} }}{ \mathrm{minimize} }  && \quad  \tau_p \sum_{l=1 }^L \sum_{k=1 }^K \hat{p}_{lk} +  (\tau_c  - \tau_p )\sum_{l=1}^L \sum_{k=1}^K p_{lk} \label{P1:a} \\
	&  \,\,\text{subject to} && \quad \mathrm{SINR}_{lk} \geq \hat{\xi}_{lk}, \quad \forall l,k, \label{P1:b} \\
	&&&  \quad \hat{p}_{lk} \leq P_{\max, lk}, \quad\forall l,k \label{P1:c},\\
	&&&  \quad p_{lk} \leq P_{\max, lk}, \quad\forall l,k \label{P1:d}.
	\end{alignat}
\end{subequations}
For a given realization of user locations, problem~\eqref{eq:Opt1} may not have an optimal solution since not all users' SE requirements can be simultaneously satisfied. Notice that only one unfortunate user with an unsatisfied SE is sufficient to create an empty feasible domain for any resource allocation problem. The problem then lacks a feasible solution \cite[Section 4.1]{Boyd2004a}.
\setcounter{eqnback}{\value{equation}} \setcounter{equation}{9}
\begin{figure*}
	\begin{align} \label{eq:ablkp}
	{I}_{lk} \left([ \hat{\mathbf{p}};\mathbf{p}] \right) = \frac{ \left( \tau_p \sum_{l' =1}^L \beta_{l'k}^l  \hat{p}_{l'k}  + \sigma^2 \right)\left( \sum_{l'=1}^L \sum_{k'=1}^K   p_{l'k'} \beta_{l'k'}^l + \sigma^2 \right)  + M \tau_p \sum_{l' = 1, l' \neq l }^L  (\beta_{l'k}^l)^2  p_{l'k} \hat{p}_{l'k} }{  M (\beta_{lk}^l)^2 \tau_p \hat{\xi}_{lk}^{-1}}
	\end{align}
	\vspace*{-0.35cm}
	\hrulefill
	\vspace*{-0.25cm}
\end{figure*}
\setcounter{eqncnt}{\value{equation}}
\setcounter{equation}{\value{eqnback}}
\subsection{Optimal Solution to Problem~\eqref{eq:Opt1}}
\vspace*{-0.15cm}
If problem~\eqref{eq:Opt1} has a nonempty feasible set, the globally optimal solution is obtained as follows.
\begin{lemma} \label{Lemma:GeometricProgram}
Problem~\eqref{eq:Opt1} can be reformulated as a geometric program and, thus, the globally optimal solution can be obtained in polynomial time if the problem is feasible.
\end{lemma}
\begin{IEEEproof}
Problem~\eqref{eq:Opt1} can be written as a geometric program on standard form \cite{Boyd2004a}. In \eqref{eq:Opt1}, the objective function is a linear combination of the uplink transmit powers, so it is a posynomial.\footnote{A function $h(x_1, \ldots, x_{N_1}) = \sum_{n=1}^{N_2}c_n \prod_{m=1}^{N_1} x_m^{a_{n,m}}$ is posynomial with $N_2$ terms ($N_2 \geq 2$) if the coefficients $a_{n,m}$  are real numbers and the coefficients $c_n$ are nonnegative real numbers. When $N_2 =1$, $h(x_1, \ldots, x_{N_1})$ is a monomial.} The power constraints are monomials, while the SINR expressions can be rearranged as posynomial constraints.
\end{IEEEproof}
Notice that geometric programs have a hidden convex structure \cite{Boyd2004a} that enables them to be efficiently solved in the centralized fashion, for example, by using the interior-point toolbox CVX \cite{cvx2015}. 

\subsection{Feasible and Infeasible Problems}
\vspace*{-0.1cm}
If the problem~\eqref{eq:Opt1} is feasible, all users will use non-zero data and pilot powers at the optimum solution since their SE requirements are assumed to be non-zero. For a given set of users, SE requirements, and power budgets, problem~\eqref{eq:Opt1} may not have a feasible solution. This can be caused by high inter-user interference, that some users might have too weak channel conditions, and/or that some users requests too high SE. If \eqref{eq:Opt1} is infeasible it cannot be solved by any method. 

However, there might still exists feasible selections of the transmit powers where most users obtain their requested SE, while only one or a few users does not. It might be enough to remove one user to turn an infeasible problem into a feasible problem, but it is nontrivial to identify which user to remove. If one attempts to solve an infeasible instance of problem~\eqref{eq:Opt1} using a standard convex optimization solver, the output might not give us any clue. The goal of this paper is to develop power control policies for such infeasible situations.

	\begin{figure*}
	\setcounter{eqnback}{\value{equation}} \setcounter{equation}{13}
	\begin{equation} \label{eq:Subtraction}
	\alpha \alpha' I_{lk}\left([ \hat{\mathbf{p}};\mathbf{p}] \right) - I_{lk}\left([ \alpha \hat{\mathbf{p}}; \alpha' \mathbf{p}] \right) =  \frac{\alpha\left( \alpha' -  1\right) \sigma^2 \tau_p \sum_{l' =1}^L \beta_{l'k}^l  \hat{p}_{l'k} +  \alpha' \left( \alpha - 1 \right) \sigma^2 \sum_{l'=1}^L \sum_{k'=1}^K   p_{l'k'} \beta_{l'k'}^l + \left( \alpha \alpha' - 1 \right) \sigma^4 }{ M (\beta_{lk}^l)^2 \tau_p \hat{\xi}_{lk}^{-1}}
	\end{equation}
	\hrule
	\vspace*{-0.25cm}
	\setcounter{eqncnt}{\value{equation}}
	\setcounter{equation}{\value{eqnback}}
\end{figure*}
\section{Alternating Optimization Approach} 
\vspace*{-0.15cm}
This section proposes an algorithm to obtain a fixed-point solution to a modified version of problem~\eqref{eq:Opt1}, where the SINR constraints of users potentially causing the congestion issue are relaxed. 
\subsection{Preliminary Setting}
\vspace*{-0.15cm}
Before providing the algorithms, Definition~\ref{Def:Succeq} introduces important notation that will be used in this paper.
\begin{definition} \label{Def:Succeq}
	Let $\mathbf{a}$ and $\mathbf{a}'$ be $KL\times 1$ real vectors whose $m$-th element is denoted by $a_m$ and $a_m'$. The notation $\mathbf{a} \succeq \mathbf{a}'$ implies element-wise inequality: $a_m \geq a_m',$  $\forall m =1,\ldots, KL$.
\end{definition}
For sake of simplicity in comprehension, we first consider the case when problem~\eqref{eq:Opt1} is feasible so that all the requested SINR constraints are satisfied using power variables that satisfy the power budget. For a particular cell~$l$, the SINR constraint of user~$k$ in \eqref{P1:b} can be reformulated as a function of the pilot and data powers $\hat{p}_{lk}, p_{lk}$ as 
\begin{equation} \label{eq:InterferenceContraint}
\hat{p}_{lk} p_{lk} \geq {I}_{lk} \left( [\hat{\mathbf{p}}; \mathbf{p} ] \right),
\end{equation}
where $\mathbf{p} = [p_{11}, \ldots,p_{LK} ]^T$, $ \hat{\mathbf{p}} = [\hat{p}_{11}, \ldots, \hat{p}_{LK}]^T \in \mathbb{R}_{+}^{LK},$ and the concatenated vector $[\hat{\mathbf{p}}; \mathbf{p} ] \in \mathbb{R}_{+}^{2LK} $; ${I}_{lk} \left([\hat{\mathbf{p}}; \mathbf{p}] \right)$ is defined in \eqref{eq:ablkp} shown at the top of this page. Since we are jointly optimizing both the pilot and data powers, the  conventional power control method based on the standard interference function, for example \cite{Yates1995a}, cannot be directly applied to our framework, which handles \eqref{eq:InterferenceContraint} and attain a fixed point. Therefore, we introduce the so-called joint interference functions to enable the joint pilot and data power allocation as shown in Definition~\ref{Def:StandardFunc}.
\begin{definition} (Joint interference function) \label{Def:StandardFunc}
A function $I\left([\mathbf{a}; \mathbf{a}' ] \right)$ is a joint interference function, if the following properties are satisfied: $a)$ Positivity: $I \left([\mathbf{a}; \mathbf{a}' ] \right) > \mathbf{0}$, $\forall [\mathbf{a}; \mathbf{a}' ] \succeq \mathbf{0}$. $b)$ Monotonicity: $I\left([\tilde{\mathbf{a}}; \tilde{\mathbf{a}}' ] \right) \geq I\left([\mathbf{a}; \mathbf{a}' ] \right)$ if $[\tilde{\mathbf{a}}; \tilde{\mathbf{a}}' ] \succeq [\mathbf{a}; \mathbf{a}' ]$. $c)$ Scalability: $\alpha\alpha' I([\mathbf{a}; \mathbf{a}' ]) > I \left( [\alpha \mathbf{a}; \alpha' \mathbf{a}'] \right)$, for all constants $\alpha, \alpha' >1$.
\end{definition}
The positivity property comes from the inherent interference and noise in radio systems. Consequently, the transmit powers cannot be zero if users have non-zero SE requirements. The monotonicity property implies that the joint interference function can be scaled up or down by adjusting the power vector. Different from \cite{Yates1995a}, the scalability property in Definition~\ref{Def:StandardFunc} provides a method to uniformly scale down  the joint interference function with a factor $\alpha \alpha'$ when updating the product of the pilot and data powers $\hat{p}_{lk}p_{lk}$ from any initial selection of $\{ \hat{p}_{lk}, p_{lk} \}$ in the feasible domain. 
\begin{lemma} \label{lemma:StandardFuc}
	Assume that each base station serves at least one user, each interference function $I_{lk} \left( [ \hat{\mathbf{p}}; \mathbf{p}] \right), \forall l,k,$ defined in \eqref{eq:ablkp} is a joint interference function.
\end{lemma}
\begin{proof}
	The proof will testify that all the functions $I_{lk} \left( [\hat{\mathbf{p}}; \mathbf{p} ] \right)$ satisfy the properties in Definition~\ref{Def:StandardFunc}. Specifically, the positivity holds since 
	\setcounter{eqnback}{\value{equation}} \setcounter{equation}{10}
	\begin{equation}
	I_{lk} \left( [\hat{\mathbf{p}}; \mathbf{p} ] \right) \stackrel{(a)}{\geq}  \frac{\hat{\xi}_{lk}\sigma^4}{ M (\beta_{lk}^l)^2 \tau_p} >0, 
	\end{equation}
	where the equality in $(a)$ holds when $\hat{p}_{lk} = p_{lk} = 0, \forall l,k$. Moreover, the numerator of $I_{lk} \left( [\hat{\mathbf{p}}; \mathbf{p} ] \right)$ is an increasing function of the optimization variables $\hat{p}_{lk}$ and $p_{lk}, \forall l,k$, thus for two given concatenated vectors $[ \hat{\mathbf{p}}_1; \mathbf{p}_1]$ and $[ \hat{\mathbf{p}}_2; \mathbf{p}_2]$ with  $[ \hat{\mathbf{p}}_1; \mathbf{p}_1] \succeq [ \hat{\mathbf{p}}_2; \mathbf{p}_2]$, it holds that
	\begin{equation}
	 I_{lk} \left([ \hat{\mathbf{p}}_1; \mathbf{p}_1] \right) \geq I_{lk} \left([ \hat{\mathbf{p}}_2; \mathbf{p}_2] \right),
	\end{equation}
	and therefore the monotonicity is proved. Finally, the scalability holds since
	\begin{equation}
	 \alpha \alpha' I_{lk} \left( [\hat{\mathbf{p}}; \mathbf{p}] \right) - I_{lk}( [\alpha \hat{\mathbf{p}}; \alpha' \mathbf{p}]) > 0, \forall \alpha, \alpha' > 1,
	\end{equation}
	as shown in \eqref{eq:Subtraction} at the top of this page, with noticing that $\alpha \alpha' \geq 1$. This completes the proof.
\end{proof}

\begin{lemma}
Assume that $0 \leq I_{lk}\left([ \hat{\mathbf{p}}; \mathbf{p}] \right) \leq P_{\max,lk}^2$ is always true for $\forall \hat{\mathbf{p}}, \mathbf{p}$ in the feasible domain. For a given concatenated vector $[\hat{\mathbf{p}}(0); \mathbf{p}(0)]$ with $\hat{p}_{lk} (0) = \hat{p}_{lk} (0) = P_{\max,lk}, \forall l,k$, there exist pilot and power coefficients such that $I_{lk}\left([ \hat{\mathbf{p}} (n); \mathbf{p} (n)] \right)$ is a non-increasing function of the power vectors along iteration index, which converges to a fixed point.
\end{lemma}
\begin{proof}
The proof is done by induction. We start updating the powers with $\hat{\mathbf{p}}(0), \mathbf{p}(0)$ and iteration~$n$ updates the pilot and data power variables as
\setcounter{eqnback}{\value{equation}} \setcounter{equation}{14}
\begin{equation} \label{eq:Updatepphat}
\hat{p}_{lk} (n) p_{lk} (n) = {I}_{lk} \left( [ \hat{\mathbf{p}} (n-1), \mathbf{p} (n-1)] \right).
\end{equation}
The range of $I_{lk}\left([ \hat{\mathbf{p}}; \mathbf{p}]\right)$ as in the lemma ensures that there exists a concatenated vector $[ \hat{\mathbf{p}}(1), \mathbf{p}(1)]$ satisfies \eqref{eq:Updatepphat} and 
\begin{equation} \label{eq:Condition}
[\hat{\mathbf{p}}(0),\mathbf{p}(0)]  \succeq [\hat{\mathbf{p}}(1),\mathbf{p}(1)].
\end{equation}
We suppose that the fact \eqref{eq:Condition} holds up to iteration~$n$, i.e., $[\hat{\mathbf{p}} (n-1); \mathbf{p} (n-1)] \succeq [\hat{\mathbf{p}} (n); \mathbf{p} (n)]$, then we must prove the existence of a vector $[\hat{\mathbf{p}} (n+1); \mathbf{p} (n+1)]$ and $[\hat{\mathbf{p}} (n); \mathbf{p} (n)] \succeq [\hat{\mathbf{p}} (n+1); \mathbf{p} (n+1)]$, for which $I_{lk} \left( [\hat{\mathbf{p}} (n-1); \mathbf{p} (n-1)] \right) \geq I_{lk} \left( [\hat{\mathbf{p}} (n); \mathbf{p} (n)] \right)$. Indeed, the following series of the inequalities holds
\begin{equation} \label{eq:Seriesinequality}
\begin{split}
&\hat{p}_{lk} (n) p_{lk} (n) \stackrel{(a)}{=} {I}_{lk} \left( [\hat{\mathbf{p}} (n-1);\mathbf{p} (n-1)]  \right) \\
& \stackrel{(b)}{\geq}  {I}_{lk} \left( [ \hat{\mathbf{p}} (n); \mathbf{p} (n)]  \right) \stackrel{(c)}{=} \hat{p}_{lk} (n+1) p_{lk} (n+1),
\end{split}
\end{equation}
where $(a)$ and $(c)$ are based on the update in \eqref{eq:Updatepphat}, while $(b)$ is based on the monotonicity. From \eqref{eq:Seriesinequality}, the bounds $\hat{p}_{lk} (n) p_{lk} (n) \geq \hat{p}_{lk} (n+1) p_{lk} (n+1), \forall l,k,$ give the existence of $[\hat{\mathbf{p}} (n+1); \mathbf{p} (n+1)]$ and we complete the proof.
\end{proof}
Every user~$k$ in cell~$l$ has a joint interference function $I_{lk} \left( [\hat{\mathbf{p}}; \mathbf{p}] \right)$ satisfying the properties in Definition~\ref{Def:StandardFunc}, which makes sure that we can construct an iterative algorithm that converges to a fixed-point solution to the initial point of pilot and data powers as explained in the next subsection.

\subsection{Solution to Problem~\eqref{eq:Opt1} based on Alternating Optimization}
\vspace*{-0.2cm}
Notice that the aforementioned analysis works on the assumption that problem~\eqref{eq:Opt1} is feasible. The power budget constraints  \eqref{P1:c} and \eqref{P1:d} are handled by observing the obvious convergence of an update $\hat{p}_{lk} (n) p_{lk} (n) = P_{\max,lk}^2, \forall n$ leading to the selection $\hat{p}_{lk} (n) = p_{lk} (n) = P_{\max,lk}, \forall n$. We then define the constrained joint interference function for user~$k$ in cell~$l$ as
\begin{equation} \label{eq:constrainedIlk}
\begin{split}
 &\hat{I}_{lk} \left( [\hat{\mathbf{p}} (n-1); \mathbf{p} (n-1)] \right) \\
 & = \min \left( I_{lk} \left( [\hat{\mathbf{p}} (n-1); \mathbf{p} (n-1)] \right), P_{\max,lk}^2 \right).
\end{split}
\end{equation}
For a cellular Massive MIMO system with the interference constraints in \eqref{eq:InterferenceContraint} and the initial power vector $\mathbf{p}(0)$ with $p_{lk}(0) = \hat{p}_{lk}(0) = P_{\max, lk}, \forall l,k,$ iteration~$n$ can update
\begin{equation} \label{eq:UpdateN}
\hat{p}_{lk} (n) p_{lk} (n) = \hat{I}_{lk} \left( [\hat{\mathbf{p}} (n-1); \mathbf{p} (n-1)] \right).
\end{equation}
From \eqref{eq:constrainedIlk} and \eqref{eq:UpdateN}, if $\hat{I}_{lk} \left( [\hat{\mathbf{p}} (n-1); \mathbf{p} (n-1)] \right) =  P_{\max,lk}^2$,  the current pilot and data powers are updated at iteration~$n$ as
\begin{equation}
\hat{p}_{lk} (n) =  p_{lk} (n) = P_{\max, lk} ,
\end{equation}
which maintains the non-increasing property of the objective function in \eqref{P1:a}. Otherwise, it holds that $\hat{I}_{lk} \left( [\hat{\mathbf{p}} (n-1); \mathbf{p} (n-1)] \right) = I_{lk} \left( [\hat{\mathbf{p}} (n-1); \mathbf{p} (n-1)] \right)$ and we must find a proper solution to $\hat{p}_{lk} (n)$ and $p_{lk} (n)$ that also satisfies the power budget constraints \eqref{P1:c} and \eqref{P1:d}. Mathematically, the solution to the pilot and data powers of user~$k$ in cell~$l$, which fulfills the above requirements, is attained by solving the following optimization problem
\begin{equation} \label{eq:v1}
\begin{aligned}
 &  \underset{ \substack{\hat{p}_{lk}, p_{lk} \in  \mathbb{R}_{+} } }{\mathrm{minimize}}
& & \tau_p  \hat{p}_{lk} +  (\tau_c - \tau_p) p_{lk} \\
& \,\,\textrm{subject to}
&&  \hat{p}_{lk} p_{lk} = {I}_{lk} ( [\hat{\mathbf{p}} (n-1); \mathbf{p} (n-1)]), \\
&&& \tau_p  \hat{p}_{lk} +  (\tau_c - \tau_p) p_{lk} \leq  t(n-1), \\
&&&  \hat{p}_{lk} \leq P_{\mathrm{max}, lk },\\
&&& p_{lk} \leq  P_{\mathrm{max}, lk },
\end{aligned}
\end{equation}
where $ t(n-1) = \tau_p  \hat{p}_{lk}(n-1) +  (\tau_c - \tau_p) p_{lk}(n-1)$.
\begin{algorithm}[t]
	\caption{Joint pilot and data power control to problem~\eqref{eq:Opt1} by alternating optimization} \label{Algorithm1}
	\textbf{Input}:  Define maximum powers $P_{\max,lk}, \forall l,k$; Select initial values $\hat{p}_{lk}(0) = p_{lk}(0) = P_{\max,lk}, \forall l,k$; Compute the energy $E_{l}(0), \forall l,$ using \eqref{eq:Etransl}; Set initial value $n=1$ and tolerance $\epsilon$.
	\begin{itemize}
			\item[1.] User~$k$ in cell~$l$ computes the joint interference function $	{I}_{lk} \left([ \hat{\mathbf{p}} (n);\mathbf{p} (n)] \right)$ using \eqref{eq:ablkp}.
			\item[2.] If ${I}_{lk} \left([ \hat{\mathbf{p}} (n);\mathbf{p} (n)] \right) > P_{\max,lk}^2$, update $\hat{p}_{lk}(n) =  p_{lk}(n) = P_{\max,lk}$. Otherwise, update $\hat{p}_{lk}(n), p_{lk}(n)$ by solving problem \eqref{eq:v1}.
			\item[3.] Repeat Steps $1,2$ with other users, then compute the ratio \fontsize{9}{9}{$\gamma (n) =$ $| \sum_{l=1}^L E_{l}(n) - E_{l}(n-1) | /  \sum_{l=1}^L \ E_{l}(n-1)$}.
			\item[4.] If $\gamma_l (n) \leq \epsilon$ $\rightarrow$ Set $\hat{p}_{lk}^{\ast} = \hat{p}_{lk}(n),$ $p_{lk}^{\ast} = p_{lk}(n),$ and Stop. Otherwise, set $n= n+1$ and go to Step $1$.
	\end{itemize}
	\textbf{Output}: A fixed point $ \hat{p}_{lk}^{\ast}$ and $ p_{lk}^{\ast}$, $\forall l,k$. \vspace*{-0.2cms}
\end{algorithm}
Since problem~\eqref{eq:v1} is convex with only the two optimization variables, the global optimal solution can be obtained easily by using CVX with an interior-point solver or, alternatively, by a two-dimensional grid search with a given accuracy. Furthermore, this problem only minimizes the total energy consumption of a specific user~$k$ in cell~$l$, thus it basically provides a local solution of $\hat{p}_{lk}$ and $p_{lk} $ to problem~\eqref{eq:Opt1}. The proposed alternating approach obtaining a fixed point to problem~\eqref{eq:Opt1} is summarized in Algorithm~\ref{Algorithm1}. The convergence of this algorithm can be proved by utilizing the methodology in \cite[Theorem~7]{Yates1995a}. We stress that the proposed algorithm can be applied to both feasible and infeasible systems. 
\begin{figure*}
	\setcounter{eqnback}{\value{equation}} \setcounter{equation}{22}
	\begin{align} \label{eq:ablkd}
	{I}_{dlk} \left(\mathbf{p} \right) =\ \frac{ \left( \tau_p \sum_{l' =1}^L \beta_{l'k}^l  \hat{p}_{l'k}  + \sigma^2 \right)\left( \sum_{l'=1}^L \sum_{k'=1}^K   p_{l'k'} \beta_{l'k'}^l + \sigma^2 \right)  + M \tau_p \sum_{l' = 1, l' \neq l }^L  (\beta_{l'k}^l)^2  p_{l'k} \hat{p}_{l'k} }{  M (\beta_{lk}^l)^2 \tau_p \hat{p}_{lk} \hat{\xi}_{lk}^{-1}}
	\end{align} \vspace*{-0.2cm}
	\hrule
	\setcounter{eqncnt}{\value{equation}}
	\setcounter{equation}{\value{eqnback}}
	\vspace*{-0.25cm}
\end{figure*}
\begin{remark}
Algorithm~\ref{Algorithm1} does not explicitly exploit the hidden convex structure for a low complexity algorithm design and problem~\eqref{eq:v1} only focuses on minimizing the total energy consumption for a particular user each time, thus a fixed-point solution to the pilot and data powers is not identical to the globally optimal solution in general. However, the total energy consumption between them has small difference as shown by numerical results in Section~\ref{Sec:NumericalResults}.
\end{remark}

\subsection{Uplink Energy Minimization for Data Power Control Only}
\vspace*{-0.1cm}
For sake of completeness, we also study the case of which only the data powers are optimized. If the feasible set is not empty, the optimization structure is given in Corollary~\ref{Co:PilotorData}.
\begin{corollary} \label{Co:PilotorData}
	If the system only optimizes the data powers, problem~\eqref{eq:Opt1} reduces to a linear program, for which the constraints~\eqref{P1:c} are removed.
\end{corollary}
Corollary~\ref{Co:PilotorData} shows that when optimizing the data powers, problem~\eqref{eq:Opt1} becomes a convex problem on standard form, thus the use of an alternating approach is guaranteed to obtain the globally optimal solution. Moreover, to perform the power control for both feasible and infeasible systems, we utilize the standard interference function from \cite{Yates1995a}.
\begin{definition} (Standard interference function) \label{Def:StandardFuncvd}
	A function $I\left(\mathbf{a} \right)$ is a standard interference function, if the following properties are satisfied: $a)$ Positivity: $I \left(\mathbf{a} \right) >  \mathbf{0}, \forall  \mathbf{a} \succeq \mathbf{0}$. $b)$ Monotonicity: $I\left( \tilde{\mathbf{a}} \right) > I\left(\mathbf{a} \right)$ if $\tilde{\mathbf{a}} \succeq \mathbf{a}$. $c)$ Scalability: $\alpha I(\mathbf{a}) > I \left( \alpha \mathbf{a}\right)$, for all constants $\alpha >1$.
\end{definition}
By removing the fixed terms representing the energy consumption of the uplink pilot training, the optimization problem is defined to minimize the energy consumption of the uplink data transmission as
\begin{equation} \label{eq:DataOptimization}
\begin{aligned}
& \underset{ \substack{\{ p_{lk} \in  \mathbb{R}_{+} \} } }{\textrm{minimize}}
& & \sum_{l=1}^L \sum_{k=1}^K (\tau_c  - \tau_p ) p_{lk}   \\
& \,\,\textrm{subject to}
& &  R_{lk} \geq \xi_{lk},\quad \forall l,k ,\\
&&& p_{lk} \leq  P_{\mathrm{max}, lk }, \quad \forall l,k.
\end{aligned}
\end{equation}
More precisely, from the SINR constraint of user~$k$ in cell~$l$, we can define the corresponding standard interference function $I_{dlk} (\mathbf{p})$ as shown in \eqref{eq:ablkd} at the top of this page, which satisfies all the properties in Definition~\ref{Def:StandardFuncvd}. From the initial setup $p_{lk}(0) = P_{\max,lk}, \forall l,k,$ the proposed algorithm will then update the data power at iteration~$n$ as
\setcounter{eqnback}{\value{equation}} \setcounter{equation}{23}
\begin{equation} \label{eq:plkn}
p_{lk} (n) = \min \left( I_{dlk} (\mathbf{p} (n-1)), P_{\max,l,k} \right), \forall l,k,
\end{equation}
and this iterative process will converge to a fixed-point solution. The proposed iterative solver is shown in Algorithm~\ref{Algorithm2}. This algorithm yields the global solution if problem~\eqref{eq:DataOptimization} is feasible.

\begin{algorithm}[t]
	\caption{Data power control to problem~\eqref{eq:DataOptimization} by alternating optimization} \label{Algorithm2}
	\textbf{Input}:  Setup maximum values $P_{\max,lk}, \forall l,k$; Select initial values $\hat{p}_{lk} = p_{lk}(0) = P_{\max,lk}, \forall l,k,$; Compute the energy of uplink data transmission $E(0) = \sum_{l=1}^L \sum_{k=1}^K (\tau_c -\tau_p) p_{l,k}(0)$; Set initial value $n=1$ and tolerance $\epsilon$.
	\begin{itemize}
		\item[1.] User~$k$ in cell~$l$ updates the data power $	p_{lk} (n)$ using \eqref{eq:plkn}.
		\item[2.] Repeat Step $1$ with other users, then update $E(n) = \sum_{l=1}^L \sum_{k=1}^K (\tau_c -\tau_p) p_{l,k}(n)$ and compute the ratio \fontsize{9}{9}{$\gamma (n) =$ $| E(n) - E(n-1) | / E(n-1)$}.
		\item[3.] If $\gamma (n) \leq \epsilon$ $\rightarrow$ Set $p_{lk}^{\ast} = p_{lk}(n),$ and Stop. Otherwise, set $n= n+1$ and go to Step $1$.
	\end{itemize}
	\textbf{Output}: A fixed point $ p_{lk}^{\ast}$, $\forall l,k$.
\end{algorithm}
\begin{figure*}[t]
	\begin{minipage}{0.48\textwidth}
		\centering
		\includegraphics[trim=0.5cm 0cm 0.7cm 0.5cm, clip=true, width=2.8in]{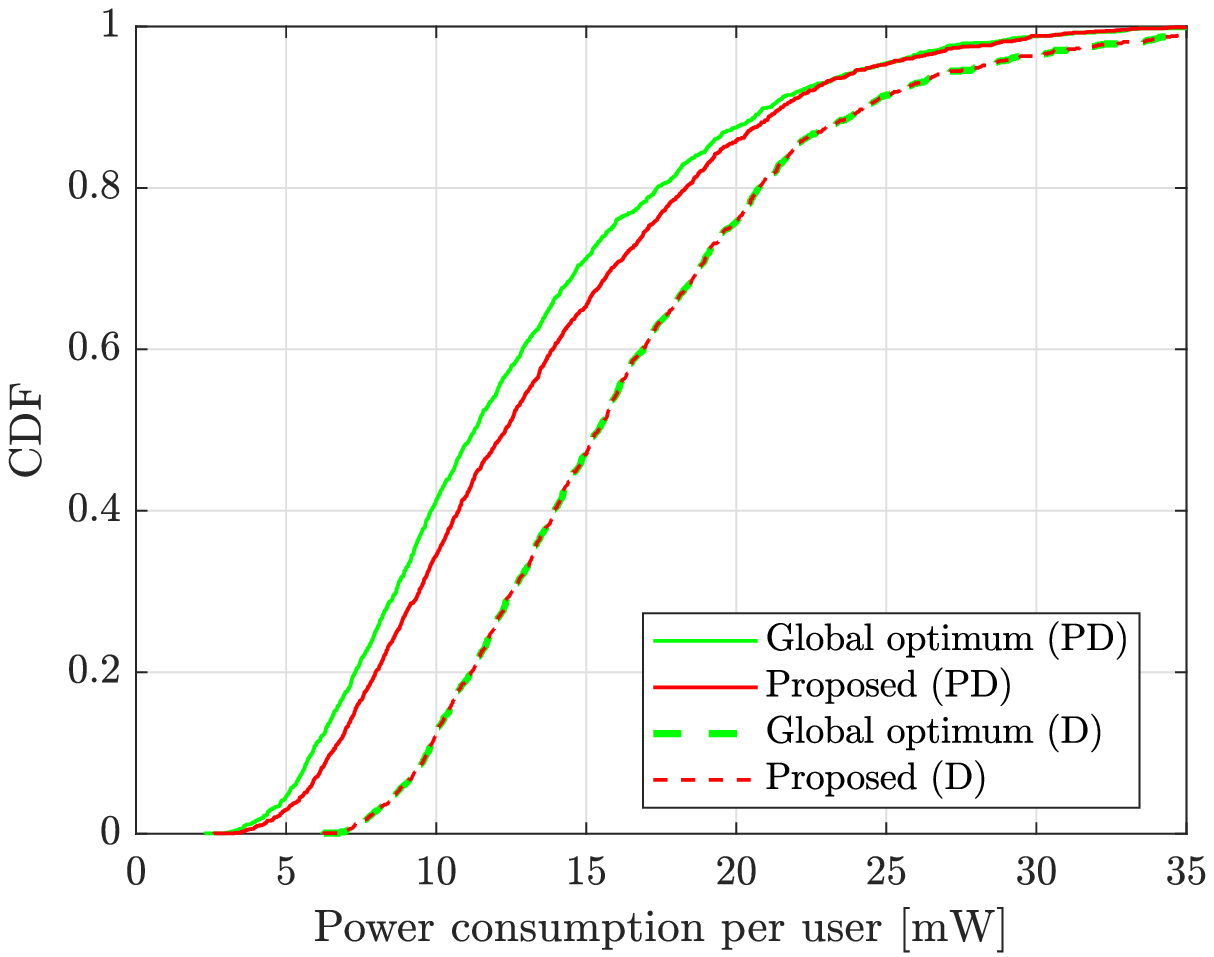} \vspace*{-0.3cm}
		\caption{The CDF of the power consumption per user [mW] for feasible systems. (PD) is the joint pilot and data power control and (D) is the data power control only. }
		\label{Fig-Antennas}
		\vspace*{-0.2cm}
	\end{minipage}
	\hfill
	\begin{minipage}{0.48\textwidth}
		\centering
		\includegraphics[trim=0.5cm 0cm 0.7cm 0.5cm, clip=true, width=2.8in]{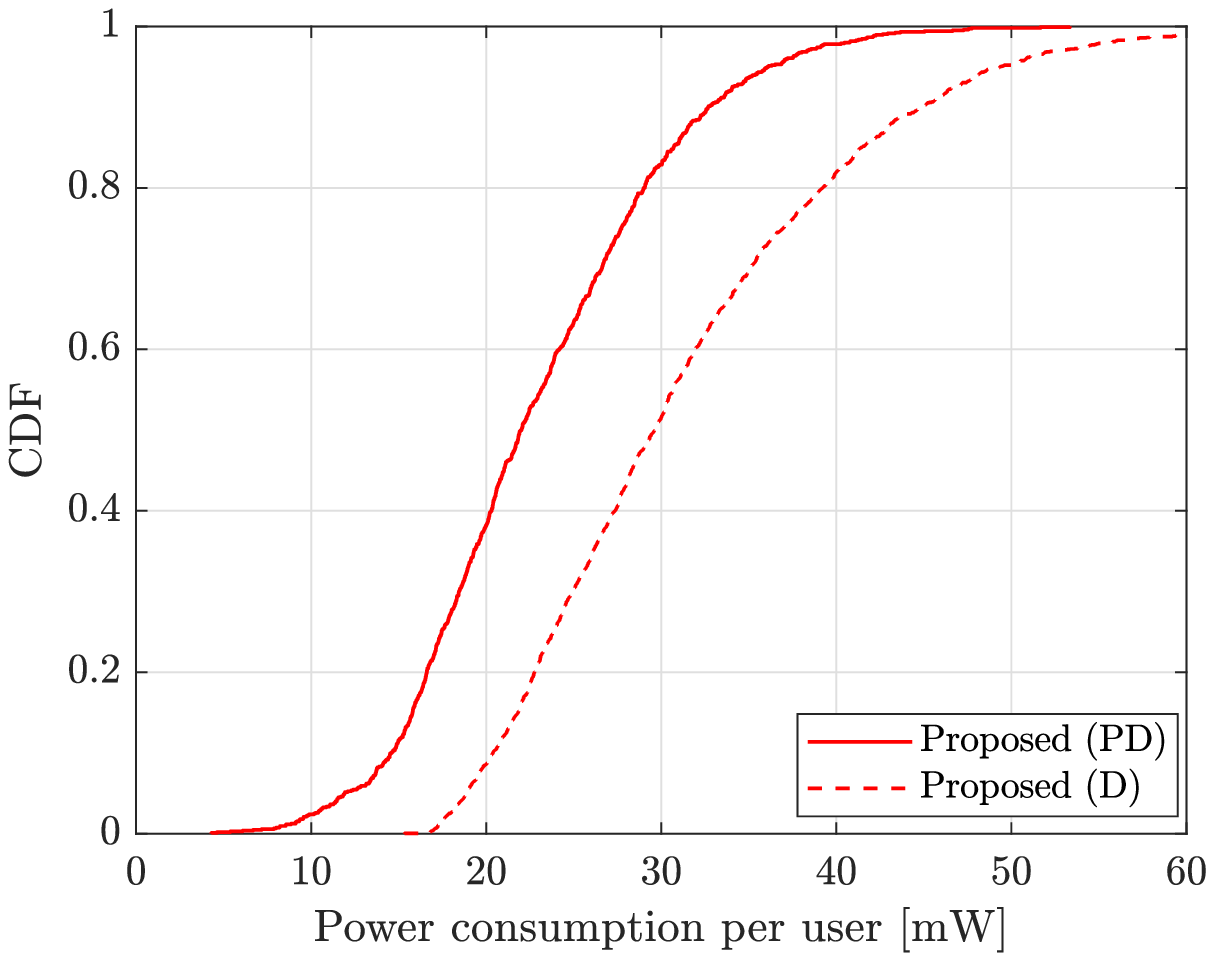} \vspace*{-0.3cm}
		\caption{The CDF of power consumption per user [mW] for infeasible systems. (PD) is the joint pilot and data power control and (D) is the data power control only.}
		\label{Fig-QoS}
		\vspace*{-0.2cm}
	\end{minipage}
\vspace*{-0.2cm}
\end{figure*}
\begin{figure}[t]
	\centering
	\includegraphics[trim=0.5cm 0cm 0.7cm 0.5cm, clip=true, width=2.8in]{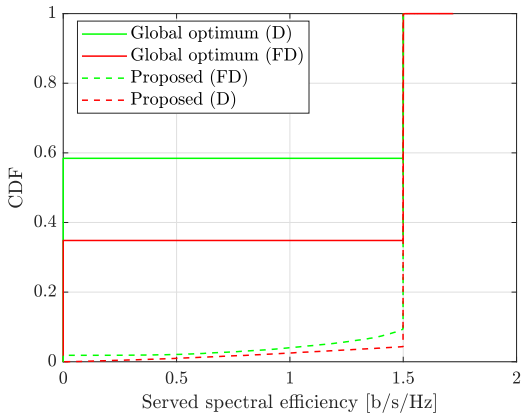} \vspace*{-0.3cm}
	\caption{The CDF of served SE per user [b/s/Hz] in both feasible and infeasible domain. (PD) is the joint pilot and data power control and (D) is the data power control only.}
	\label{FigServedSE}
	\vspace*{-0.2cm}
\end{figure}
\section{Numerical Results} \label{Sec:NumericalResults}
\vspace*{-0.1cm}
A cellular Massive MIMO system with $L=4$ square cells in an area $1$~km$^2$ is considered.  Each cell has a BS located at the center, serving  $K=5$ users which are uniformly distributed within its cell and we assume that no user is closer
to its BS than $35$~m. The required SE of each user is $1.5$ [b/s/Hz]. The wrap-around technique is applied to avoid boundary effects. We model the large-scale fading coefficients based on the 3GPP LTE specifications \cite{LTE2010a}. The system uses $20$~MHz of bandwidth, and the noise variance is $-96$~dBm with the noise figure $5$~dB. The large-scale fading coefficient $\beta_{l'k'}^l$ [dB] is
\begin{equation}
\beta_{l'k'}^l = -148.1 - 37.6 \log_{10} \left( d_{l'k'}^l / 1 \mathrm{km} \right) + z_{l'k'}^l,
\end{equation}
where $d_{l'k'}^l \geq 35$~m denotes the distance between user~$k'$ in cell~$l'$ and BS~$l$. The shadow fading coefficient $z_{l'k'}^l$ has a Gaussian distribution with zero mean and standard derivation $7$~dB. The pilot and data symbols have a maximum power of $200$~mW. Monte-Carlo simulations are done over $3000$ random sets of user locations. In the proposed algorithms, we set $\epsilon = 0.01$. We include the global optimum (obtained using interior point methods)  from previous works \cite{Chien2016b, senel2017joint} as reference, but only for feasible systems.

Fig.~\ref{Fig-Antennas} shows the cumulative distribution function (CDF) of the uplink power per user for feasible systems. The joint pilot and data power control requires the lowest power per user. The global optimum requires $12.6$~mW per user on average, while the proposed method needs $13.3$~mW  ($6\%$ more). Hence, the proposed algorithm does not find the global optimum. Fig.~\ref{Fig-Antennas} also demonstrates the behavior when we only optimize the pilot powers. In this case, the proposed method finds the global optimum, but it requires on average $27\%$ more power than when optimizing both pilot and data powers. 

Fig.~\ref{Fig-QoS} show the CDF of the power consumption [mW] per user for infeasible systems. This is the case of main interest in this paper since there is no global optimum to compute or compare against. As can be seen from the figure, all users transmit with non-zero power at the operating points identified by the proposed algorithms. By jointly optimizing the pilot and data powers, we can reduce the power consumption by $36\%$ on the average.

Fig.~\ref{FigServedSE} shows the CDF of the actual SE per user [b/s/Hz] by utilizing the different power control algorithms with the requested SE $1.5$ [b/s/Hz]. We compare our proposed algorithm, which guarantees non-zero SE, with the solution obtained with the convex optimization solvers. In the latter case, all the SEs are set to zero for infeasible setups.
We notice that $34.8\%$ of the considered setups are infeasible when performing joint pilot and data power control, while $58.5\%$ are infeasible when performing only data power control. In contrast, when using the proposed methods, we can satisfy the SE requirements of most users also in the infeasible setups. More precisely, only $5\%$ or $9\%$ of users, respectively, cannot receive their requested SE in those cases, and even those users obtain non-zero SE.

\vspace*{-0.2cm}
\section{Conclusion}
\vspace*{-0.15cm}
There are known algorithms that can minimize the transmit power required to deliver the SE required by the users in cellular Massive MIMO systems, but these cannot be applied when the requirements cannot be simultaneously achieved.
This paper has developed a joint pilot and data power control algorithm that can deal with such situations. Instead of trying to actively identify the problematic users and removing them from service, we develop an algorithm where most users get their requested SEs while the problematic users are given a lower SE than requested, which can still be non-zero. The numerical results demonstrated that in scenarios where many random user realizations lead to infeasible power control problems, we can still satisfy the SE requirements for the vast majority of the users.

\vspace*{-0.4cm}
\bibliographystyle{IEEEtran}
\bibliography{IEEEabrv,refs}
\end{document}